\documentclass[final,3p,times]{elsarticle}
\usepackage{graphicx,epsfig}
\DeclareGraphicsExtensions{.eps}
\usepackage{amsfonts}
\usepackage{amssymb}
\usepackage{amsmath}

\usepackage{hyperref}

\newtheorem{theorem}{Theorem}[section]

\newtheorem{proposition}[theorem]{Proposition}

\newenvironment{proof}[1][Proof]{\begin{trivlist}
		\item[\hskip \labelsep {\bfseries #1}]}{\end{trivlist}}

\newcommand{\diag}{\mathop{\mathrm{diag}}} 

\usepackage{url}

\makeatletter
    \def\ps@pprintTitle{%
      \let\@oddhead\@empty
      \let\@evenhead\@empty
      \let\@oddfoot\@empty
      \let\@evenfoot\@oddfoot
    }
\makeatother

\long\def\symbolfootnote[#1]#2{\begingroup%
\def\thefootnote{\fnsymbol{footnote}}\footnote[#1]{#2}\endgroup}

\begin{document}

\begin{frontmatter}

\title{Structure preserving algorithms for simulation of linearly damped acoustic systems\symbolfootnote[2]
{NOTICE: The archived file is not the final published version of the article: V. Chatziioannou, Structure preserving algorithms for simulation of linearly damped acoustic systems, in: ``Journal of Numerical Analysis, Industrial and Applied Mathematics (JNAIAM)'', Vol. 13 (2019) 33--54, \copyright European Society of Computational Methods in Sciences and Engineering. The definitive publisher-authenticated version is available online at: \url{http://jnaiam.org/exit.php?url_id=239&entry_id=141}}}
\author{Vasileios Chatziioannou\corref{cor1}}\cortext[cor1]{Corresponding author~Tel.: +43 1 71155 4313.}
\ead{chatziioannou@mdw.ac.at}
\address{Department of Music Acoustics (IWK), University of Music and performing Arts Vienna, Anton-von-Webern-Platz 1, 1030 Vienna, Austria\corref{cor2}}


%
%
%
\begin{abstract}

Energy methods for constructing time-stepping algorithms are of increased interest in application to nonlinear problems, since numerical stability can be inferred from the conservation of the system energy. Alternatively, symplectic integrators may be constructed that preserve the symplectic form of the system. This methodology has been established for Hamiltonian systems, with numerous applications in engineering problems. In this paper an extension of such methods to non-conservative acoustic systems is presented. Discrete conservation laws, equivalent to that of energy-conserving schemes, are derived for systems with linear damping, incorporating the action of external forces. Furthermore the evolution of the symplectic structure is analysed in the continuous and the discrete case. Existing methods are examined and novel methods are designed using a lumped oscillator as an elemental model. The proposed methodology is extended to the case of distributed systems and exemplified through a case study of a vibrating string bouncing against a rigid obstacle.

\end{abstract}
\end{frontmatter}

\section{Introduction}\label{sec;intro}

Time-stepping methods have seen increased attention in the numerical simulation of mechanical systems since the continuous advance in computer hardware allows to analyse systems of increasing complexity. Of particular interest, when employing such methods for computer simulation, is the stability of the numerical algorithms. Especially in the case of nonlinear systems special care needs to be taken in order to ensure a bound on the underlying model variables.

To this cause energy preserving schemes have been developed for the discretisation of ordinary and partial differential equations \cite{greenspan,hairer_EC10,li95_energy,macias2018pde,quispel08} and systems thereof \cite{bacchini2018ode,stefan_strings,gonzalez96}. In the field of music acoustics, from where the case studies treated in this paper are drawn, conservation of numerical energy has been the central theme in recent implementation designs involving nonlinear interactions \cite{stefan_aca14,chabassier_cmame,jsv}, especially when non-analytic forces are present. Numerical stability is thus ensured by the conservation of the system energy (or an energy-like positive definite quantity).

The particular requirements of physical modelling sound synthesis applications dictate the use of simple, low-order integration methods. This is due to perceptual considerations \cite{stefan_NSS}, as well as uncertainties in the input parameters of the physical models (especially in the case of material parameters and damping factors).
Furthermore, the aim to employ the constructed models in real-time synthesis applications requires the use of computationally efficient models with proven stability properties. Nevertheless, a brief discussion will also follow on the construction of higher-order methods that share the same preservation properties with the presented methods. Finally, the requirement for full audio-bandwidth applications and the necessity to couple the derived numerical schemes to other digital representations of sounding objects poses the restriction of a constant sampling rate (usually the audio sampling rate of 44.1 kHz). For this reason, methods involving variable step-sizes (see, e.g.\ \cite{brugnano_var_dt,hairer_EC10}) are rarely considered in Music Acoustics applications.

In this paper the problematic of energy conserving schemes is transferred to systems with linear damping. When frictional forces are present, the system energy is not constant any more and alternative conserved quantities need to be identified in order to evaluate the behaviour of numerical simulations \cite{celledoni_AVF,laburta15}. Since dissipative terms are expected to assist towards solutions not blowing up in finite time, stability is of a lesser concern here, compared to the case of conservative (undamped) systems. Of primary interest is to obtain numerical algorithms that can preserve structural invariants inherent to the differential equations. That is, besides obtaining a guarantee for numerical stability, the existence of conserved numerical quantities, corresponding to an analytical counterpart, may serve as ``\emph{a criterion to judge the success of a numerical simulation}'' \cite{li95_energy}.

To this cause, an extensive analysis is initially presented on the damped harmonic oscillator in Section \ref{sec;lumped}. Conservation laws are derived both in the continuous and in the discrete domain and several discretisation methods are discussed. Their stability is analysed in terms of the evolution of the energy and the symplectic structure of the system. It is shown that the impulse invariance method \cite{pollock} (frequently used in digital signal processing applications) and the Caldirola-Kanai approach \cite{caldirola,kanai} (used in quantum mechanics) may lead to the design of conformally symplectic mappings \cite{wojtkowski}. A way to visualise such mappings (following the treatment of symplectic mappings in the case of conservative systems) is presented. Section \ref{sec;nl} applies the proposed methodology to the case of nonlinear systems and introduces the incorporation of external forces to the presented approach, including a case study of a clarinet reed simulation. Section \ref{sec;dist} presents an extension to the case of infinite dimensional systems, focusing on the simulation of a string-barrier collision and Section \ref{sec;conc} summarises the main findings of this work.

\subsection{Definitions}\label{sec;def}

Consider a system of Hamiltonian ordinary differential equations
\begin{equation}\label{HODE}
\dfrac{\textrm{d}\mathbf{y}}{\textrm{d}t}=\dfrac{\partial H}{\partial \mathbf{p}},\qquad \dfrac{\textrm{d}\mathbf{p}}{\textrm{d}t}=-\dfrac{\partial H}{\partial \mathbf{y}}
\end{equation}
where $\mathbf{y},\mathbf{p}\in\mathbb{R}^d$, and $H(\mathbf{y},\mathbf{p})$ is the Hamiltonian function, corresponding to the total energy of the system. This defines a conservative system, where $H(\mathbf{y}(t),\mathbf{p}(t)) = H(\mathbf{y}(t_0),\mathbf{p}(t_0))$. The solution of this system can be described by the flow
\begin{equation}\label{flow0}
\phi_t(\mathbf{y}(t_0),\mathbf{p}(t_0))=(\mathbf{y}(t_0+t),\mathbf{p}(t_0+t))
\end{equation}
which defines a symplectic transformation, i.e.\ it preserves the symplectic structure $\omega=\textrm{d}y\wedge\textrm{d}p$:
\begin{equation}\label{pullback_def}
\phi_t^{\boldsymbol{*}}\omega = \omega
\end{equation}
where $\phi_t^{\boldsymbol{*}}$ denotes the pull-back of $\omega$ by $\phi_t$ (see, e.g., \cite{arnold78,mratiu}). In geometric terms, this corresponds to the preservation of all the Poincar{\'e} integral invariants \cite{serna}, leading to the condition $\det(\phi_t')=1$, where $\phi_t'$ is the Jacobian of $\phi_t$. Consequently, the Hamiltonian vector field $\mathbf{v_H}=(\mathbf{\dot y},\mathbf{\dot p})$ is divergence free (with $\mathbf{\nabla\cdot v_H}=0$).

A numerical approximation to the solution of (\ref{HODE}) at time $t^n=n\Delta t$ is given by $(\mathbf{y}^n,\mathbf{p}^n)$, with $\Delta t$ being the sampling interval. Hence a (one-step) numerical integrator of (\ref{HODE}) generates a discrete mapping $\Phi_{\Delta t}:\mathbb{R}^{2d}\mapsto\mathbb{R}^{2d}$ with
\begin{equation}\label{map0}
(\mathbf{y}^{n+1},\mathbf{p}^{n+1}) = \Phi_{\Delta t}(\mathbf{y}^n,\mathbf{p}^n).
\end{equation}
Energy conserving schemes should obey $H^{n+1} = H^n$
where $H^n = H(\mathbf{y}^{n},\mathbf{p}^{n})$. For symplectic mappings $\textrm{d}\mathbf{y}^{n+1}\wedge\textrm{d}\mathbf{p}^{n+1} =\textrm{d}\mathbf{y}^{n}\wedge\textrm{d}\mathbf{p}^{n}$
should hold. A \emph{mechanical integrator}, according to the definition given by Wendlandt and Marsden \cite{wendlandt}, has to preserve either the energy or the symplectic form, while it has been shown that, in general, both invariants can not be preserved \cite{zhong}. (Note, however, that for a given problem it is possible to construct symplectic methods that also conserve energy, as explained in \cite{brugnano_ECsymp}, by fine-tuning a discretisation parameter at each time-step.)

For the construction of numerical integrators it is helpful to define the following difference operators, acting on a discrete time approximation $\chi^n$ of a continuous variable $\chi$
\begin{subequations}\label{deltas0}
	\begin{align}\label{deltas}
		\delta_{t+}\chi^n = \dfrac{\chi^{n+1}-\chi^n}{\Delta t},&\qquad \delta_{t-}\chi^n = \dfrac{\chi^{n}-\chi^{n-1}}{\Delta t}.
	\end{align}
	Similarly averaging operators are defined as
	\begin{align}\label{mus}
		\mu_{t+}\chi^n = \dfrac{\chi^{n+1}+\chi^n}{2},&\qquad \mu_{t-}\chi^n = \dfrac{\chi^{n}+\chi^{n-1}}{2}.
	\end{align}
\end{subequations}

\section{The damped harmonic oscillator}\label{sec;lumped}

\subsection{Continuous domain}\label{sec;cont}

The equation of motion for the displacement $y$ of a damped harmonic oscillator is given by

\begin{equation}\label{dhomot0}
\dfrac{\textrm{d}^2y}{\textrm{d}t^2} + \gamma\dfrac{\textrm{d}y}{\textrm{d}t} + \omega_0^2 y = 0
\end{equation}
where $\gamma$ is the damping and $\omega_0$ the resonance frequency of the oscillator. Multiplying by the mass $m$ of the oscillator yields
\begin{equation}\label{dhomot}
m\dfrac{\textrm{d}^2y}{\textrm{d}t^2} + m\gamma\dfrac{\textrm{d}y}{\textrm{d}t} + ky = 0
\end{equation}
where $k=m\omega_0^2$ is the stiffness. Two initial conditions have to be specified for this second order equation, namely $y(0)$ and $\dot y(0)$, where the dot signifies differentiation with respect to time. This dissipative system can be written in Hamiltonian form (see, e.g., \cite{dressler_lyapunov}) as
\begin{equation}\label{dhoHeq}
\dfrac{\textrm{d}y}{\textrm{d}t} = \dfrac{\partial H}{\partial p},\qquad\dfrac{\textrm{d}p}{\textrm{d}t} = -\dfrac{\partial H}{\partial y} -\gamma p
\end{equation}
where $H(y,p)=T(p)+V(y)$ is the sum of the kinetic energy $T$ and the potential energy $V$, with $T(p)=p^2/(2m)$ and $V(y)=ky^2/2$, $p=\partial L/\partial\dot y$ being the conjugate momentum, where $L=T-V$ is the Lagrangian of the system. The required initial condition here is $(y(0),p(0))$. In the undamped case (for $\gamma=0$) this is clearly a Hamiltonian system, which conserves the total energy $H$ and the symplectic structure $\omega = \textrm{d}y\wedge \textrm{d}p$ \cite{arnold78}. In the presence of damping, energy is dissipated according to
\begin{equation}\label{dHdt}
\dfrac{\textrm{d}H}{\textrm{d}t} = \dfrac{\partial H}{\partial y}\dfrac{\textrm{d}y}{\textrm{d}t} +\dfrac{\partial H}{\partial p}\dfrac{\textrm{d}p}{\textrm{d}t} = -\dfrac{\gamma p^2}{m}\leq 0
\end{equation}
which induces the following conservation law
\begin{equation}\label{Hcon}
H+\int\dfrac{\gamma p^2}{m}\;dt = \mbox{const}.
\end{equation}
Similar to the system energy, the symplectic area is also dissipated. In fact \cite{mc_dyn}, it is preserved up to a multiple $c(t)$, that is
\begin{equation}\label{pullback}
\phi_t^{\boldsymbol{*}}\omega = e^{-\gamma t}\omega
\end{equation}
where $\phi_t$ is the time-$t$ flow of the conformal vector field $\mathbf{v_{H,\gamma}} = (\dot y,\dot p)$ and $c(t) = e^{-\gamma t}$. The divergence of this field is
\begin{equation}\label{divg}
\textrm{div}(\mathbf{v_{H,\gamma}})=\mathbf{\nabla\cdot v_{H,\gamma}} = -\gamma.
\end{equation}
Note that $\phi$ lies in the \emph{conformal symplectic group} ${\textrm{Diff}\,}_\omega^c$ of diffeomorphisms that preserve a symplectic 2-form $\omega$ up to a factor \cite{mc_conf}. The above quantities (energy and symplectic area) will be used for the analysis of different numerical schemes.

For $\gamma/2<\omega_0$ the system admits an oscillatory solution. Taking the Laplace transform of (\ref{dhomot0}) yields the characteristic equation
\begin{equation}\label{cheq}
\sigma^2 + \gamma \sigma + \omega_0^2 = 0
\end{equation}
which is solved by $\sigma=\gamma/2 \pm \textrm{j}\,\omega_\gamma$, with $\omega_\gamma=\sqrt{\omega_0^2-(\gamma/2)^2}$ being the frequency of the damped oscillator. The exact solution can then be written as
\begin{equation}\label{anal}
y_\textrm{an}(t) = Ae^{-\gamma t/2}\cos(\omega_\gamma t + \theta)
\end{equation}
where the subscript `\emph{an}' stands for \emph{analytical}, with
\begin{subequations}\label{anal2}
	\begin{equation}
	\dot y_\textrm{an}(t) = -A\omega_\gamma e^{-\gamma t/2}\sin(\omega_\gamma t + \theta) - \frac{\gamma}{2}Ae^{-\gamma t/2}\cos(\omega_\gamma t + \theta)
	\end{equation}
	and
	\begin{equation}
	H_\textrm{an}(t) = \frac{m}{2}\dot y_\textrm{an}^2+\frac{k}{2}y_\textrm{an}^2
	\end{equation}
\end{subequations}
where the amplitude $A$ and the phase $\theta$ of the oscillation can be obtained from the initial conditions. For the overdamped case (when $\gamma/2>\omega_0$; not treated in this paper) a solution of the form $y_\textrm{an}(t) = c_1e^{\sigma_{+}t} + c_2e^{\sigma_{-}t}$, with $\sigma_\pm = -\gamma/2 \pm\sqrt{(\gamma/2)^2-\omega_0^2}$ is obtained, with the system exponentially decaying to its equilibrium position.

\subsection{Discretisation domain}\label{sec;disc}

A series of numerical schemes will be evaluated, where, in terms of structure preservation, it will be examined whether they respect the numerical counterpart of the evolution of energy and symplectic structure.
The discrete energy of the system is defined as
\begin{equation}\label{Hdisc}
H^n = \frac{(p^n)^2}{2m} + \frac{k}{2}(y^n)^2
\end{equation}
In order to examine the transition from $H^n$ to $H^{n+1}$ (\ref{dHdt}) can be discretised using the forward difference operator $\delta_{t+}H^n$ and the time-averaging operator $\mu_{t+}p^n$ as
\begin{equation}\label{dHdtdisc}
\dfrac{H^{n+1}-H^n}{\Delta t} = -\dfrac{\gamma}{m}\left(\dfrac{p^{n+1}+p^n}{2}\right)^2.
\end{equation}
Note that the averaging of the discrete momentum needs to take place, so that both sides of the equation are centred around time $(n+1/2)\Delta t$. This leads to the discrete conservation law
\begin{equation}\label{Hcond}
K^{n} = H^{n+1} + \sum_{\kappa=0}^n\dfrac{\gamma}{m}\left(\mu_{t+}p^\kappa\right)^2\Delta t = \mbox{const}.
\end{equation}
which is the discrete equivalent of (\ref{Hcon}). The error in the conservation of this quantity is measured using the preservation metric
\begin{equation}\label{metricK}
\displaystyle{\mathcal{K} = \dfrac{1}{N+1}\sum_{n=0}^N\dfrac{K^{n+1}-K^n}{K^0}}
\end{equation}
where $N$ is the number of time-steps taken. It corresponds to the average deviation per sample of the conserved quantity $K$ normalised with respect to $K^0$. (Note that different numerical approximations may be used to evaluate (\ref{dHdt}), but these may result in larger computational stencils or non-centred approximations.) 

The symplectic structure $\omega$ evolves subject to the transformation $(y^n,p^n)\mapsto(y^{n+1},p^{n+1})$ that corresponds to the chosen numerical integration algorithm, that is \cite{serna}
\begin{equation}\label{dynp1}
\textrm{d}y^{n+1}\wedge \textrm{d}p^{n+1} = \left(\dfrac{\partial y^{n+1}}{\partial y^n}\dfrac{\partial p^{n+1}}{\partial p^n} -\dfrac{\partial y^{n+1}}{\partial p^n}\dfrac{\partial p^{n+1}}{\partial y^n}\right)\textrm{d}y^n\wedge \textrm{d}p^n:=\mathcal{D}\;\textrm{d}y^n\wedge \textrm{d}p^n.
\end{equation}
The contraction relation for conformal symplectic mappings can be written as (see \cite{mc_conf})
\begin{equation}\label{De}
\textrm{d}y^{n+1}\wedge \textrm{d}p^{n+1} = e^{-\gamma\Delta t}\;\textrm{d}y^n\wedge \textrm{d}p^n
\end{equation}
hence the equality $\mathcal{D}=e^{-\gamma\Delta t}$ must hold in order for the symplectic form to contract exactly at the correct rate. A \emph{mechanical integrator}, as defined in Section~\ref{sec;intro} for conservative systems, should here either obey the conservation law (\ref{Hcond}) or contract the symplectic form according to (\ref{De}).
It should be noted that, from the above definitions, the conformal symplectic character of the evolution is more general; the derivation of (\ref{Hcond}) depends on the choice of the finite difference operators, whereas the derivation of (\ref{De}) is universal.

The accuracy of the various numerical approximations will be also judged by comparison with the energy of the system, as calculated using the analytic solution (\ref{anal}).
The presence of such an exact solution allows a direct evaluation of the accuracy of various numerical schemes, something that can not always be achieved when integrating nonlinear systems (as in Section~\ref{sec;nl}). If the time series $\mathbf{H}_{\Delta t} = \big(H^1,H^2,\ldots,H^N\big)$ is derived from an approximate solution, with $H^n = H(y^n,p^n)$, then the deviation of the approximate energy from the exact value is calculated using the following metric
\begin{equation}\label{metrics}
H_\textrm{dev} = 100\dfrac{\|\mathbf{H}_{\Delta t}-\mathbf{H}_\textrm{an}\|_2}{\overline{\mathbf{H}_\textrm{an}}}
\end{equation}
where $\overline{\mathbf{H}_\textrm{an}}$ is the mean value of $\mathbf{H}_\textrm{an} = \big(H_\textrm{an}(\Delta t),H_\textrm{an}(2\Delta t),\ldots,H_\textrm{an}(N\Delta t)\big)$.

\subsection{A divergence-free field}\label{sec;divw}

In the case of conservative systems, the fact that the symplectic form remains constant is visualised by means of a divergence-free vector field. In order to visualise how the contraction of the symplectic form is respected for dissipative systems, a divergence-free field $\mathbf{w}$ is defined. This lies on the modified phase space
\begin{equation}\label{MTR}
\mathcal{W} = e^{\gamma t/2}T^{\boldsymbol{*}}Q=(\psi,~\xi) = \left\{\Big(e^{\gamma t/2}y(t),~e^{\gamma t/2}p(t)\Big)\, : \,y\in Q,p\in T_y^{\boldsymbol{*}}Q\right\}
\end{equation}
where $Q$ is the configuration space of the system (here $Q=\mathbb{R}$), $T_y^{\boldsymbol{*}}Q$ the cotangent space of $Q$ at $y$ and $T^{\boldsymbol{*}}Q$ the cotangent bundle of $Q$.

\begin{theorem}
	For every vector field $\mathbf{v} = (\dot y,\dot p)$ whose time-$t$ flow lies in the conformal symplectic group ${\mathrm{Diff}\,}_\omega^c$, there exists a divergence-free field $\mathbf{w}$ defined on the modified phase space $\mathcal{W}$.
\end{theorem}
\begin{proof}
	Let $\mathbf{w}= (\dot\psi,\dot\xi)$ with $\psi = e^{\gamma t/2}y(t)$ and $\xi = e^{\gamma t/2}p(t)$. The divergence of this field is
	\begin{equation*}
		\textrm{div}(\mathbf{w}) = \mathbf{\nabla\cdot w} = \left(\frac{\partial}{\partial y},\frac{\partial}{\partial p}\right)\cdot \left(\frac{\gamma}{2}e^{\gamma t/2}y +e^{\gamma t/2}\dot y,~\frac{\gamma}{2}e^{\gamma t/2}p + e^{\gamma t/2}\dot p\right)
	\end{equation*}
	Now, since the flow of $\mathbf{v}$ lies in ${\textrm{Diff}\,}_\omega^c$, equation (\ref{dhoHeq}) holds, hence
	\begin{equation*}
		\textrm{div}(\mathbf{w}) = \frac{\partial}{\partial y}\left(\frac{\gamma}{2}e^{\gamma t/2}y +e^{\gamma t/2}\frac{\partial H}{\partial p}\right) + \frac{\partial}{\partial p}\left(\frac{\gamma}{2}e^{\gamma t/2}p + e^{\gamma t/2}\Big(-\frac{\partial H}{\partial y} -\gamma p\Big)\right) = 0
	\end{equation*}
	and $\mathbf{w}$ is divergence-free. $\square$
\end{proof}

As such, plotting the solution trajectory in the modified phase space $\mathcal{W}$ resembles the phase-space trajectory of a conservative system. Figure \ref{fig;omega} shows such trajectories for the numerical methods treated in this paper. This reflects how accurately a method contracts the symplectic structure, in comparison to the analytic solution. The contraction relation (\ref{De}), which can be written as
\begin{equation}\label{Demod}
e^{\gamma\Delta t/2}\textrm{d}y^{n+1}\wedge e^{\gamma\Delta t/2}\textrm{d}p^{n+1} = \textrm{d}y^n\wedge \textrm{d}p^n
\end{equation}
is visually stretched to mirror that of a conservative, divergence-free field, with
\begin{equation}\label{Depismu}
\textrm{d}\psi^{n+1}\wedge\textrm{d}\xi^{n+1} = \textrm{d}y^n\wedge \textrm{d}p^n.
\end{equation}

\subsection{Hamiltonian integrators}\label{sec;hamint}

An energy-conserving scheme (EC)\textemdash{}in the sense of equation~(\ref{Hcond})\textemdash{}whose properties have been recently demonstrated for a class of nonlinear Hamiltonian systems \cite{smac}, can be obtained by applying mid-point derivative approximations to (\ref{dhoHeq})
\begin{subequations}\label{Heqdisc}
	\begin{align}
		\dfrac{y^{n+1}-y^n}{\Delta t} &= \dfrac{T(p^{n+1})-T(p^n)}{p^{n+1}-p^n}\label{Heqdisc1}\\
		\dfrac{p^{n+1}-p^n}{\Delta t} &= -\dfrac{V(y^{n+1})-V(y^n)}{y^{n+1}-y^n} - \gamma\dfrac{p^{n+1}+p^n}{2}\label{Heqdisc2}
	\end{align}
\end{subequations}
leading to the following numerical scheme
\begin{subequations}\label{ECmap}
	\begin{align}
		p^{n+1} &= \dfrac{1-k\Delta t^2/4m -\gamma\Delta t/2}{1+k\Delta t^2/4m +\gamma\Delta t/2}p^n - \dfrac{k\Delta t}{1+k\Delta t^2/4m+\gamma\Delta t/2}y^n\\
		y^{n+1} &= y^n + \dfrac{\Delta t}{2m}(p^{n+1}+p^n).
	\end{align}
\end{subequations}
As explained in \cite{jsv}, in the case of linear systems this is equivalent to both the (symplectic) midpoint rule (which is a second order Runge-Kutta method) and the trapezoidal rule (which belongs to the family of Newmark methods). The superiority of this algorithm, in terms of energy conservation, becomes apparent when nonlinear forces act on the system (see Section~\ref{sec;nl} and Ref. \cite{jsv}).

\begin{proposition}
	The (EC) scheme defined by (\ref{Heqdisc}) exactly replicates the numerical energy balance (\ref{dHdtdisc}).
\end{proposition}
\begin{proof}
	Multiplying (\ref{Heqdisc1}) by $p^{n+1}-p^n$ and (\ref{Heqdisc2}) by $y^{n+1}-y^n$ and substituting by parts yields
	\begin{equation*}
		T(p^{n+1}) + V(y^{n+1}) = T(p^n) + V(y^n) - \gamma\frac{p^{n+1} + p^n}{2}(y^{n+1} - y^n).
	\end{equation*}
	with $(y^{n+1} - y^n)/\Delta t = (p^{n+1} + p^n)/(2m)$, hence
	\begin{equation*}
		H^{n+1} = H^n - \frac{\gamma}{m}(\mu_{t+}p^n)^2\Delta t
	\end{equation*}
	which replicates (\ref{dHdtdisc}) exactly, rendering (EC) a mechanical integrator. $\square$
\end{proof}

Note that this derivation holds for any potential function $V$ and hence also applies to nonlinear systems, with $K$ conserved to machine precision in implementations on digital processors. For the symplectic structure the following relation can be shown
\begin{equation}\label{wEC}
\textrm{d}y^{n+1}\wedge \textrm{d}p^{n+1} = \dfrac{\omega_0^2\Delta t/2 + 2 -\gamma\Delta t}{\omega_0^2\Delta t/2 + 2 +\gamma\Delta t}\;\textrm{d}y^n\wedge \textrm{d}p^n
\end{equation}
hence the mapping does not exactly replicate the conformal symplectic dynamics of the continuous system.

A well known symplectic integrator for Hamiltonian systems is given by the velocity Verlet algorithm (VV) \cite{hairerVV}. Defining $f^n=f(y^n)$ as the discretisation of the force $f=-\partial V/\partial y$ acting on the system (in this case $f(y^n)=-ky^n$), the following algorithm is obtained for the dissipative system (\ref{dhoHeq})
\begin{equation}\label{VVl}
\begin{split}
&p^{n+1/2} = \dfrac{p^n+\dfrac{\Delta t}{2}f^n}{1+\gamma\dfrac{\Delta t}{2}},\\
&y^{n+1} = y^n + \dfrac{\Delta t}{m}p^{n+1/2},\qquad p^{n+1} = \left(1-\dfrac{\Delta t}{2}\gamma\right)p^{n+1/2} + \dfrac{\Delta t}{2}f^{n+1}
\end{split}
\end{equation}
with the sympletic structure $\omega$ evolving according to
\begin{equation}\label{wVV}
\textrm{d}y^{n+1}\wedge \textrm{d}p^{n+1} = \dfrac{2-\gamma\Delta t}{2+\gamma\Delta t}\;\textrm{d}y^n\wedge \textrm{d}p^n.
\end{equation}
Thus the contraction factor $\mathcal{D}$ of the symplectic structure is the $(1,1)$ Pad\'e approximation to $e^{-\gamma\Delta t}$ \cite{pade}.

A special treatment of damped oscillators can be achieved using the Caldirola-Kanai Lagrangian \cite{caldirola,kanai}, given by
\begin{equation}\label{LCK}
L_{\textrm{CK}} = e^{\gamma t}(\frac{1}{2}m\dot y^2 -\frac{1}{2}ky^2) = e^{\gamma t}(T-V).
\end{equation}
Defining $\varpi=\partial L_{\textrm{CK}}/\partial\dot y$ and taking the Legendre transformation of $L_{\textrm{CK}}$ yields the Hamiltonian
\begin{equation}\label{HCK}
H_{\textrm{CK}} = e^{-\gamma t}\dfrac{\varpi^2}{2m} + e^{\gamma t}\dfrac{k}{2}y^2.
\end{equation}
An interesting feature of this approach is that the conjugate momentum $\varpi$ is different from the kinematic momentum $p$, with $\varpi = e^{\gamma t}m\dot y$. Hamilton's equations take their classical form
\begin{equation}\label{HeqCK}
\dfrac{\textrm{d}y}{\textrm{d}t}=\dfrac{\partial H_{CK}}{\partial \varpi},\qquad \dfrac{\textrm{d}\varpi}{\textrm{d}t}=-\dfrac{\partial H_{CK}}{\partial y}
\end{equation}
and the total energy of the system is given by $H = e^{-\gamma t}H_{CK}$. Note that
\begin{equation}\label{dHdtCK}
\frac{\textrm{d}H}{\textrm{d}t} = -\gamma e^{-\gamma t}H_{CK} + e^{-\gamma t}\frac{\textrm{d}H_{CK}}{\textrm{d}t} = -2\gamma\left(e^{-\gamma t}\right)^2\frac{\varpi^2}{2m} = -\gamma\frac{p^2}{m}
\end{equation}
in accordance to (\ref{dHdt}). Discretisation of (\ref{HeqCK}) at mid-point yields the following numerical scheme (CK)
\begin{subequations}\label{HCdisc}
	\begin{align}
		\dfrac{y^{n+1}-y^n}{\Delta t} &= e^{-\gamma(n+1/2)\Delta t}\,\dfrac{\varpi^{n+1}+\varpi^n}{2m} \label{HCdisc1}\\
		\dfrac{\varpi^{n+1}-\varpi^n}{\Delta t} &= e^{\gamma(n+1/2)\Delta t}\dfrac{k}{2}(y^{n+1}+y^n) \label{HCdisc2}.
	\end{align}
	Defining
	\begin{equation}\label{qndef}
	q^n=\varpi^ne^{-\gamma n\Delta t}\Delta t/(2m)
	\end{equation}
	and
	\begin{equation}\label{xndef}
	x^n = \left(2q^ne^{-\gamma\Delta t/2} - \frac{\Delta t^2ky^n}{2m}\right)/\left(1 + \frac{\Delta t^2k}{4m}\right)
	\end{equation}
\end{subequations}
leads to the following explicit update
\begin{equation}\label{ypCK}
q^{n+1} = x^ne^{-\gamma\Delta t/2} - q^ne^{-\gamma\Delta t},\qquad y^{n+1} = y^n + x^n
\end{equation}    
whence $\varpi^{n+1}$ can be obtained as $\varpi^{n+1} = (2m/\Delta t)q^{n+1}e^{\gamma(n+1)\Delta t}$.

\begin{proposition}
	The (CK) scheme, as defined by (\ref{HCdisc}) constitutes a conformal symplectic mapping.
\end{proposition}
\begin{proof}
	For the mapping in (\ref{HCdisc}) it can be derived that $\textrm{d}y^{n+1}\wedge \textrm{d}q^{n+1} = e^{-\gamma\Delta t}\textrm{d}y^n\wedge \textrm{d}q^n$. From (\ref{qndef}) and using the fact that $\varpi^n = e^{\gamma n\Delta t}p^n$ one can write $q^n = p^n\Delta t/2m$. Then
	\begin{equation}
	\begin{split}
	\textrm{d}y^{n+1}\wedge \textrm{d}p^{n+1} &= \frac{2m}{\Delta t}\left(\textrm{d}y^{n+1}\wedge \textrm{d}q^{n+1}\right)\\[0.2cm]
	&= \frac{2m}{\Delta t}e^{-\gamma\Delta t}\textrm{d}y^{n}\wedge \textrm{d}q^{n} = e^{-\gamma\Delta t}\textrm{d}y^{n}\wedge \textrm{d}p^{n}
	\end{split}
	\end{equation}
	in agreement with the system dynamics. $\square$
\end{proof}

\subsection{The impulse invariant method}\label{sec;newton}

A discretisation method that has seen much use, especially in signal processing applications, is the \emph{impulse invariance method} (IIM). In this approach the impulse response of the system is derived and a sampled version of it is designed for the discretisation \cite{pollock}. Defining the amplification factor $z=e^{\sigma\Delta t}$, so that $y^{n+1} = zy^n$ and assuming that $y^{n+1} + a_1 y^n + a_2 y^{n-1} = 0$, the characteristic equation (\ref{cheq}) leads to
\begin{equation}\label{chariim}
z + a_1 + a_2z^{-1} = 0.
\end{equation}
Substituting the exact value of $\sigma$ (from the solution of (\ref{cheq})), so that $z = e^{(-\gamma/2 + \textrm{j}\,\omega_\gamma)\Delta t}$ and using Euler's rule yields $a_1$ and $a_2$ so that
\begin{equation}\label{iimdisc}
y^{n+1} = 2e^{-\gamma\Delta t/2}\cos(\omega_\gamma\Delta t)y^n -e^{-\gamma\Delta t}y^{n-1}
\end{equation}
where initial conditions for $y^0$ and $y^1$ are required. Stability here is ensured if the amplification factor $|z|<1$, which holds for $1/\Delta t>\omega_0$.

\begin{proposition}
	The impulse invariant method contracts the symplectic structure $\omega$ at exactly the correct rate.
\end{proposition}
\begin{proof}
	Let $\hat{y}^n = \mu_{t-}y^n$ and $\hat{p}^n = m\delta_{t-}y^n$. Then equation (\ref{iimdisc}) can be written in the form of a mapping on $\mathbb{R}^{2d}$, namely $(\hat{y}^n,\hat{p}^n)\mapsto(\hat{y}^{n+1},\hat{p}^{n+1})$ where
	\begin{subequations}
		\begin{align*}
			\hat{y}^{n+1} &= \frac{1-a_1-a_2}{2}\;\hat{y}^n + \Delta t\frac{1-a_1+a_2}{4m}\;\hat{p}^n\\
			\hat{p}^{n+1} &= -m\frac{1+a_1+a_2}{\Delta t}\;\hat{y}^n -\frac{1+a_1-a_2}{2}\;\hat{p}^n
		\end{align*}
	\end{subequations}
	with $\textrm{d}\hat{y}^{n+1}\wedge \textrm{d}\hat{p}^{n+1} = e^{-\gamma\Delta t}\;\textrm{d}\hat{y}^n\wedge \textrm{d}\hat{p}^n$
	hence the mapping is conformal symplectic. $\square$
\end{proof}

\subsection{Splitting methods}\label{sec;split}

Amongst the Hamiltonian integrators considered in Section \ref{sec;hamint}, only the (CK) method contracts the symplectic structure according to $\phi_t^{\boldsymbol{*}}\omega = e^{-\gamma t}\omega$. It is nevertheless possible to modify the other methods in order to achieve this \cite{split}. To this cause the conformal vector field $\mathbf{v_{H,\gamma}}$ is written as the sum of two fields
\begin{equation}\label{splitv}
\mathbf{v_{H,\gamma}} = \mathbf{v_{H,0}} + \mathbf{v_{C,\gamma}} \quad\Rightarrow\quad (\dot y,\dot p) = (\dfrac{\partial H}{\partial p},-\dfrac{\partial H}{\partial y}) + (0,-\gamma p)
\end{equation}
the flow $\phi_t$ being the composition of the Hamiltonian flow $\phi^{[H]}_t$ and the flow $\phi^{[\gamma]}_t$. Now, in the case of linear dissipation, the latter flow is available exactly as $\phi^{[\gamma]}_t = (y,e^{-\gamma t}p)$. Hence the total, discrete flow of the system can be expressed as $\Phi_{\Delta t} = \Phi_{\Delta t}^{[H]}\circ\phi_{\Delta t}^{[\gamma]}$, where $\Phi_{\Delta t}^{[H]}$ is the discrete flow of the conservative (undamped) Hamiltonian system (obtained using either discretisation method). Applying this transformation to the (EC) method a new set of update equations is obtained, with an inherent `conformal symplectic' property
\begin{subequations}\label{ECconf}
	\begin{align}
		(\textrm{EC}^\textrm{[cs]}):\quad p^{n+1} &= \dfrac{1-\Delta t^2 k/4m}{1+\Delta t^2 k/4m}e^{-\gamma\Delta t}p^n - \dfrac{\Delta tk}{1+\Delta t^2 k/4m}y^n\\[0.2cm]
		y^{n+1} &= y^n + \dfrac{\Delta t}{2m}\left(p^{n+1}+e^{-\gamma\Delta t}p^n\right)
	\end{align}
\end{subequations}
where now $\mathcal{D} = e^{-\gamma\Delta t}$. The same holds for the (VV) algorithm, with the update being
\begin{equation}\label{VVconf}
\begin{split}
(\textrm{VV}^\textrm{[cs]}):\quad &p^{n+1/2} = e^{-\gamma\Delta t}p^n+\dfrac{\Delta t}{2}f^n,\\[0.2cm]
&y^{n+1} = y^n + \dfrac{\Delta t}{m}p^{n+1/2},\qquad p^{n+1} = p^{n+1/2} + \dfrac{\Delta t}{2}f^{n+1}.
\end{split}
\end{equation}
Such composition methods, apart from the design of conformal symplectic integrators, can also be used to construct higher order methods \cite{hairer_EC10,leimkuhler04,split}. However such an approach is not followed here, as it is rarely relevant in physical modelling sound synthesis applications, for the reasons explained in Section~\ref{sec;intro}.

\subsection{Numerical results}

A comparison of all the above methods, based on the metric $H_\textrm{dev}$ and the discrete conservation laws as quantified by $\mathcal{K}$ and $\mathcal{D}$ is presented on Table~\ref{tab;lin}. The evolution of the energy error $K_{err}=(K^{n+1}-K^0)/K^0$ is depicted in Figure \ref{fig;Kerr} along with the phase-space trajectory as calculated using the exact solution. Figure~\ref{fig;omega} depicts the solution trajectories in the modified phase space $\mathcal{W}=(\psi,\,\xi)$, where mappings that are not conformally symplectic (with $\mathcal{D}\neq e^{-\gamma\Delta t}$) deviate from the analytical solution.

\begin{table}[!h]
	\footnotesize
	\renewcommand{\arraystretch}{1.3}
	\caption{{Properties of numerical integrators} \label{tab;lin}}
	\centering 
	\begin{tabular}{l|  r| r| c}
		\hline
		\hline
		Method & \multicolumn{1}{c|}{$H_\textrm{dev}$} & \multicolumn{1}{c|}{$\mathcal{K}$} & \multicolumn{1}{c}{$\mathcal{D}$}\\ \hline
		EC & 6.20 & $7.26\times10^{-19}$ & $(\omega_0^2\Delta t/2 + 2 -\gamma\Delta t)/(\omega_0^2\Delta t/2 + 2 +\gamma\Delta t)$\\
		VV & 6.58 & $-6.01\times10^{-5}~$ & $(2-\gamma\Delta t)/(2+\gamma\Delta t)$\\
		CK & 2.07 & $-2.45\times10^{-5}~$ & $e^{-\gamma\Delta t}$\\ \hline
		IIM & 0.98 & $-1.41\times10^{-4}~$ & $e^{-\gamma\Delta t}$\\ \hline
		EC$^{[\mbox{cs}]}$ & 30.42 & $7.40\times10^{-5}~$ & $e^{-\gamma\Delta t}$\\
		VV$^{[\mbox{cs}]}$ & 23.05 & $5.10\times10^{-5}~$ & $e^{-\gamma\Delta t}$\\
		\hline
		\hline
	\end{tabular}
\end{table}

\begin{figure}[!h]
	\centering
	\includegraphics[width=0.86\columnwidth]{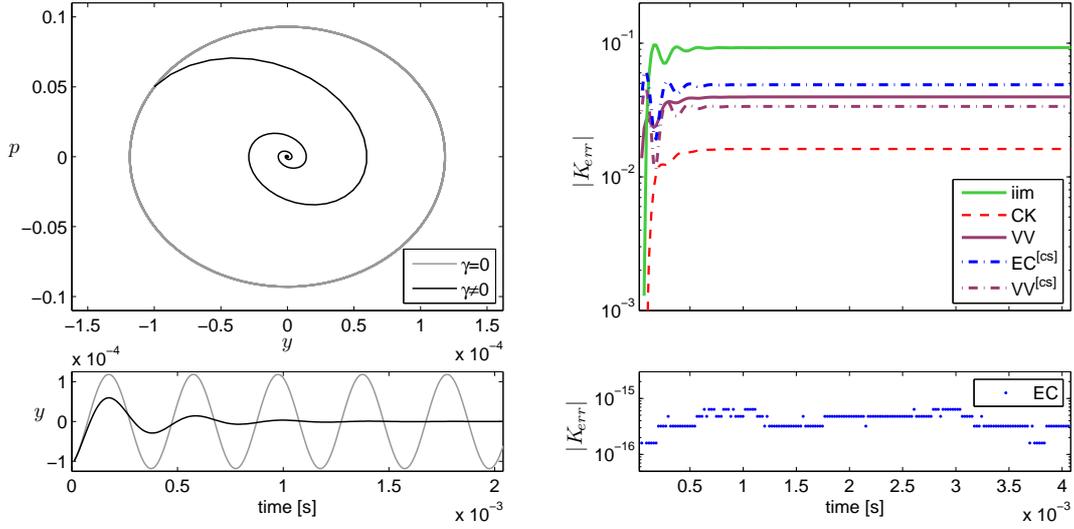}
	\caption{Left: The phase space trajectory (top) and displacement signal (bottom) calculated using the exact solution (the undamped case is plotted for comparison). Right: The energy error for all discretisation methods. Note that for the (EC) method the error remains within machine precision, exhibiting single-bit variation. \label{fig;Kerr}}
\end{figure}

\begin{figure}[!h]
	\centering
	\includegraphics[width=0.86\columnwidth]{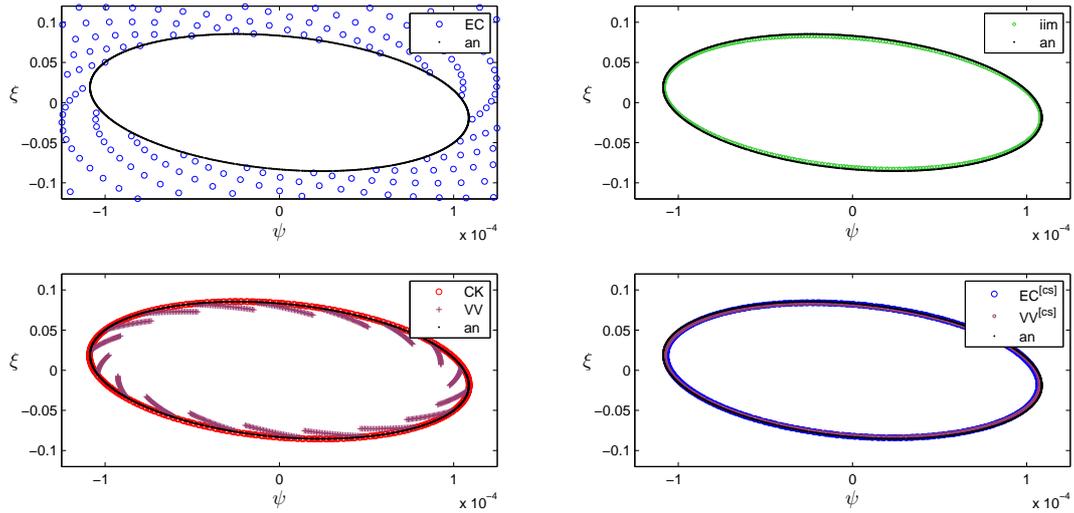}
	\caption{The trajectories of all discretisation methods and the analytical solution, plotted in the modified phase space $\mathcal{W} = (\psi,~\xi)$. \label{fig;omega}}
\end{figure}

For all simulations the physical parameters are taken from the caption of Figure~\ref{fig;ddho_NL} in Section~\ref{sec;clar}, representing a clarinet reed model (with the exemption of the damping $\gamma$, which has been increased 2.5 times to render the dissipation effect more significant). The stepsize used is $\Delta t = 1/f_s$, where $f_s = 44100$ Hz and the initial conditions are $y(0) = -0.1$ mm and $p(0) = 0.05$ kg$\,$m/s. Apart from the (EC) method, all other methods fail to respect the discrete conservation law (\ref{Hcond}). The most accurate approximations for the system energy (regarding $H_\textrm{dev}$) are obtained using the (IIM) and the (CK) methods, which are the ones that manage to replicate the symplectic dynamics of the continuous system in the discrete domain. The methods that are inherently not conformal sympletic yield worse approximations, which also deteriorate when using composition methods to recover the correct dynamics.

In the author's opinion, the advantage of the conformal symplectic methods, in terms of replicating the system dynamics, should be taken into account when designing mechanical integrators. As mentioned in Section~\ref{sec;disc} it is a more fundamental property, in comparison to the discrete conservation law, since it is uniformly defined. However such a design may not always be available for more complex problems, e.g.\ when frequency dependent damping or external forces are present. In that case preservation of a numerical energy-like quantity should be considered in order to ensure algorithm stability. Such a quantity exists, by construction, when the (EC) scheme is used, regardless of the type of the potential energy of the system. Alternative discretisation methods may yield similar results, but it is not always straightforward to identify the conserved numerical quantity.

\section{Nonlinear oscillators}\label{sec;nl}

Simulating the behaviour of a lumped oscillator becomes a more interesting problem, when nonlinear forces act on the system. In this section the effect of a non-smooth impact force is considered, that becomes active when the oscillating mass tries to exceed a certain boundary (here located at $y=0$). A common approach to simulate such forces allows a small penetration inside the `rigid' boundary, which can be equivalently considered as the compression of the impacting objects \cite{harmon,penalty}. Starting from Hertz's contact law, the impact force takes the form
\begin{equation}\label{fcoll}
f(y) = -k_c\lfloor y^\alpha\rfloor
\end{equation}
where $\lfloor y^\alpha\rfloor = h(y)\,y^\alpha$, $h(y)$ denotes the Heaviside step function, $k_c$ is a stiffness coefficient and the power law exponent $\alpha\geq1$ depends on the local shape of the contact surface \cite{papetti}. This results in the potential energy taking the form $V(y) = ky^2/2 + k_c\lfloor y^{\alpha+1}\rfloor/(\alpha+1)$ and the equation of motion for the oscillator becomes
\begin{equation}\label{dhomotNL}
m\dfrac{\textrm{d}^2y}{\textrm{d}t^2} + m\gamma\dfrac{\textrm{d}y}{\textrm{d}t} + ky + k_c\lfloor y^\alpha\rfloor= 0.
\end{equation}
This does not affect the conservation law (\ref{Hcon}), since the form of the potential energy is not used during its derivation. Hence, in the absence of an analytical solution to this problem, the numerical approximation can be assessed via the discrete conservation law (\ref{Hcond}). One should however notice that such a nonlinear potential is only $\alpha$-times differentiable at $y=0$, resulting in a decline of the accuracy of the numerical approximation, which is otherwise second-order accurate. This may lead to an energy drift (see Figure~\ref{fig;long10_per}) as explained in \cite{smith2012reflections}. However the error introduced by this effect is negligible in the systems examined in this study (namely linearly damped acoustic systems) since the energy drift is masked out by frictional losses.

The integration methods of the previous section are extended to this nonlinear problem as follows.
For the Hamiltonian integrators the updated potential energy, incorporating the nonlinear impact force, needs to be substituted in the formulation of the numerical schemes. As already mentioned in Section~\ref{sec;lumped}, in the presence of nonlinear forces, the (EC) method becomes distinct from the midpoint rule (MR) and the trapezoidal rule (TR). In all these formulations it is required to solve a nonlinear equation at each time-step, which takes the following form for each method:
\begin{align}
	& (\textrm{EC}):\quad \lambda\,\frac{V(y^n+s)-V(y^n)}{s} + (1+\gamma\Delta t/2)s - 2q^n = 0\label{nlEC}\\
	& (\textrm{MR}):\quad \lambda\,V'\Big(\frac{2y^n+s}{2}\Big) + (1+\gamma\Delta t/2)s - 2q^n = 0\\
	& (\textrm{TR}):\quad \lambda\,\frac{V'(y^n+s)+V'(y^n)}{2} + (1+\gamma\Delta t/2)s - 2q^n = 0
\end{align}
with $\lambda = \Delta t^2/(2m)$ and the unknown $s = y^{n+1} - y^{n}$. Existence and uniqueness of solutions for the above equations stem from the convexity of the potential $V$ (see \cite{jsv}). 

The same substitution of the potential energy $V$ can be applied to the (CK) method, leading to the solution of a nonlinear equation in $s$ (defined as above)
\begin{equation}\label{CKNL}
\lambda\frac{k_c}{\alpha+1}\frac{\lfloor (y^n+s)^{\alpha+1}\rfloor-\lfloor (y^n)^{\alpha+1}\rfloor}{s} + s - 2q^n e^{-\gamma\Delta t} + \frac{k\Delta t^2}{4m}(s+2y^n) = 0.
\end{equation}
For the velocity Verlet algorithm (VV) the only necessary update is applied to the force acting on the system, with $f = -ky^n - k_c\lfloor (y^n)^\alpha\rfloor$.

Similarly, for the impulse invariance method, the update equation becomes
\begin{equation}\label{iimdiscNL}
y^{n+1} = -\frac{\Delta t e^{-\gamma\Delta t/2}\sin(\omega_\gamma\Delta t)}{m\omega_\gamma}k_c\lfloor y^\alpha\rfloor + 2e^{-\gamma\Delta t/2}\cos(\omega_\gamma\Delta t)y^n -e^{-\gamma\Delta t}y^{n-1}.
\end{equation} Finally, the splitting procedure of Section~\ref{sec;split} is still applicable in the same way, regardless of the presence of nonlinear forces.

\subsection{External forces}\label{sec;clar}

Another point of interest in practical applications is the presence of external forces driving the oscillations of the system. Such a power input can be incorporated to the (continuous and discrete) conservation laws. Given an external force $f_\textrm{ex}$, Hamilton's equations take the following form
\begin{equation}\label{ddhoHeq}
\dfrac{\textrm{d}y}{\textrm{d}t} = \dfrac{\partial H}{\partial p},\qquad\dfrac{\textrm{d}p}{\textrm{d}t} = f_\textrm{ex} -\dfrac{\partial H}{\partial y} -\gamma p.
\end{equation}
with
\begin{equation}\label{dhreed}
\dfrac{\textrm{d}H}{\textrm{d}t} = -\dfrac{\gamma p^2}{m} + \dfrac{p}{m}f_\textrm{ex}.
\end{equation}
Accordingly the conservation law (\ref{Hcon}) becomes
\begin{equation}\label{dHcon}
H+\int\dfrac{\gamma p^2 - pf_\textrm{ex}}{m}\;dt = \mbox{const}.
\end{equation}
with the following discretised version 
\begin{equation}\label{dHcond}
K^{n} = H^{n+1} + \sum_{\kappa=0}^n\left(\gamma\left(\mu_{t+}p^\kappa\right)^2 - \left(\mu_{t+}p^\kappa\right)\left(\mu_{t+}f_\textrm{ex}^\kappa\right)\right)\dfrac{\Delta t}{m} = \mbox{const}.
\end{equation}
The above equations present a conserved quantity, where there is both an energy loss mechanism and a power input to the system energy $H$ due to the action of both frictional and external forces.

Applying a periodic driving force usually results in a steady-state displacement signal preceded by a transient oscillation. This is demonstrated in this section using a problem from musical acoustics. In particular, the motion of a clarinet reed is simulated and the resulting sound pressure is synthesised. The clarinet reed is driven by the pressure difference across it $p_\Delta = p_\textrm{m} - p_\textrm{in}$, where $p_\textrm{m}$ is the blowing pressure (mouthpressure) and $p_\textrm{in}$ is the pressure inside the clarinet mouthpiece (see Figure \ref{fig;ddho_NL}(b)).

\begin{figure}[!t]
	\begin{minipage}[t]{0.48\textwidth}
		\begin{center}
			(a)\\
			\includegraphics[width=.9\linewidth]{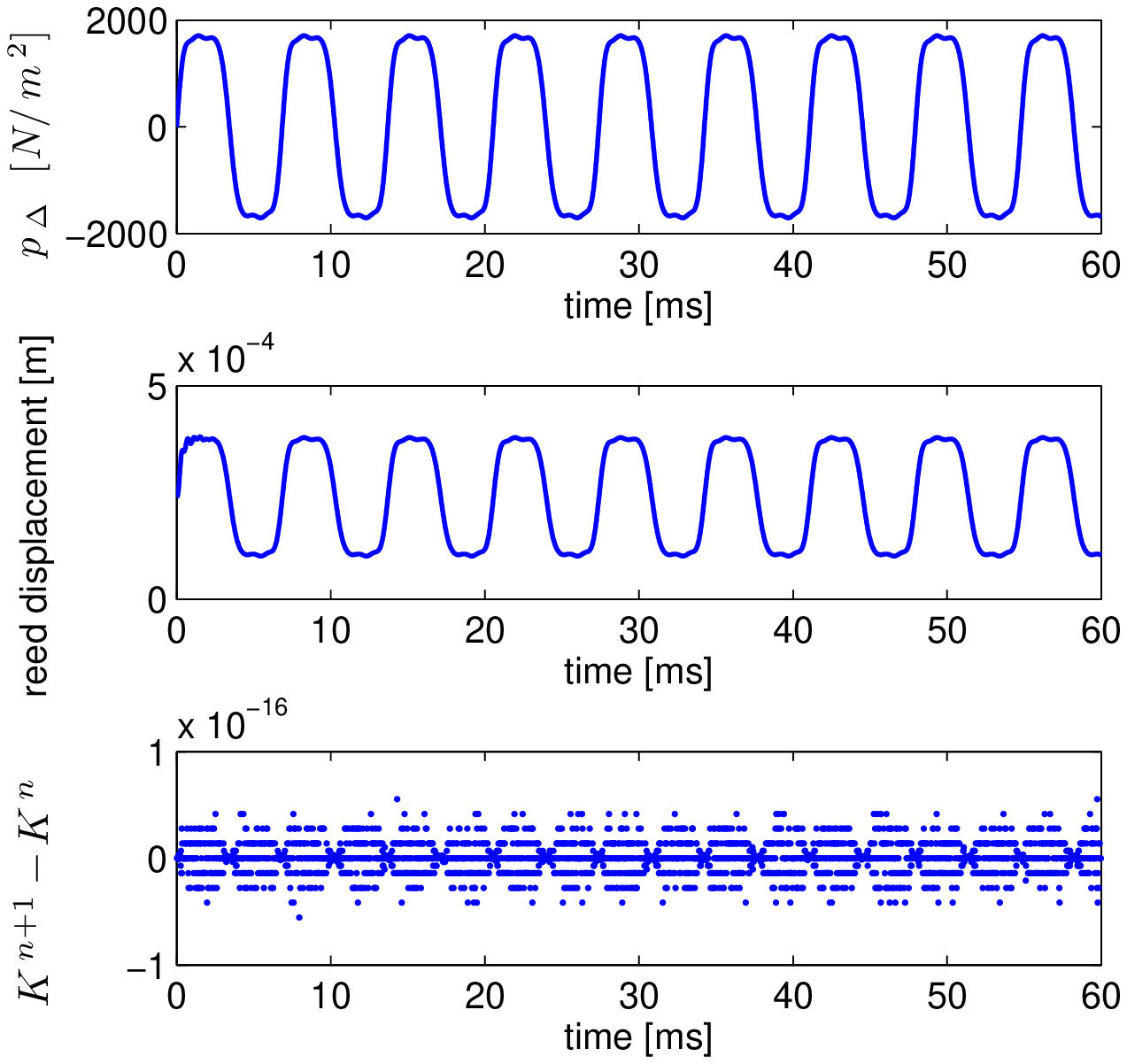}
		\end{center}
	\end{minipage}
	\hfill
	\begin{minipage}[t]{0.48\textwidth}
		\begin{center}
			\hspace{0.75cm}(b)\\
			\includegraphics[width=.9\linewidth]{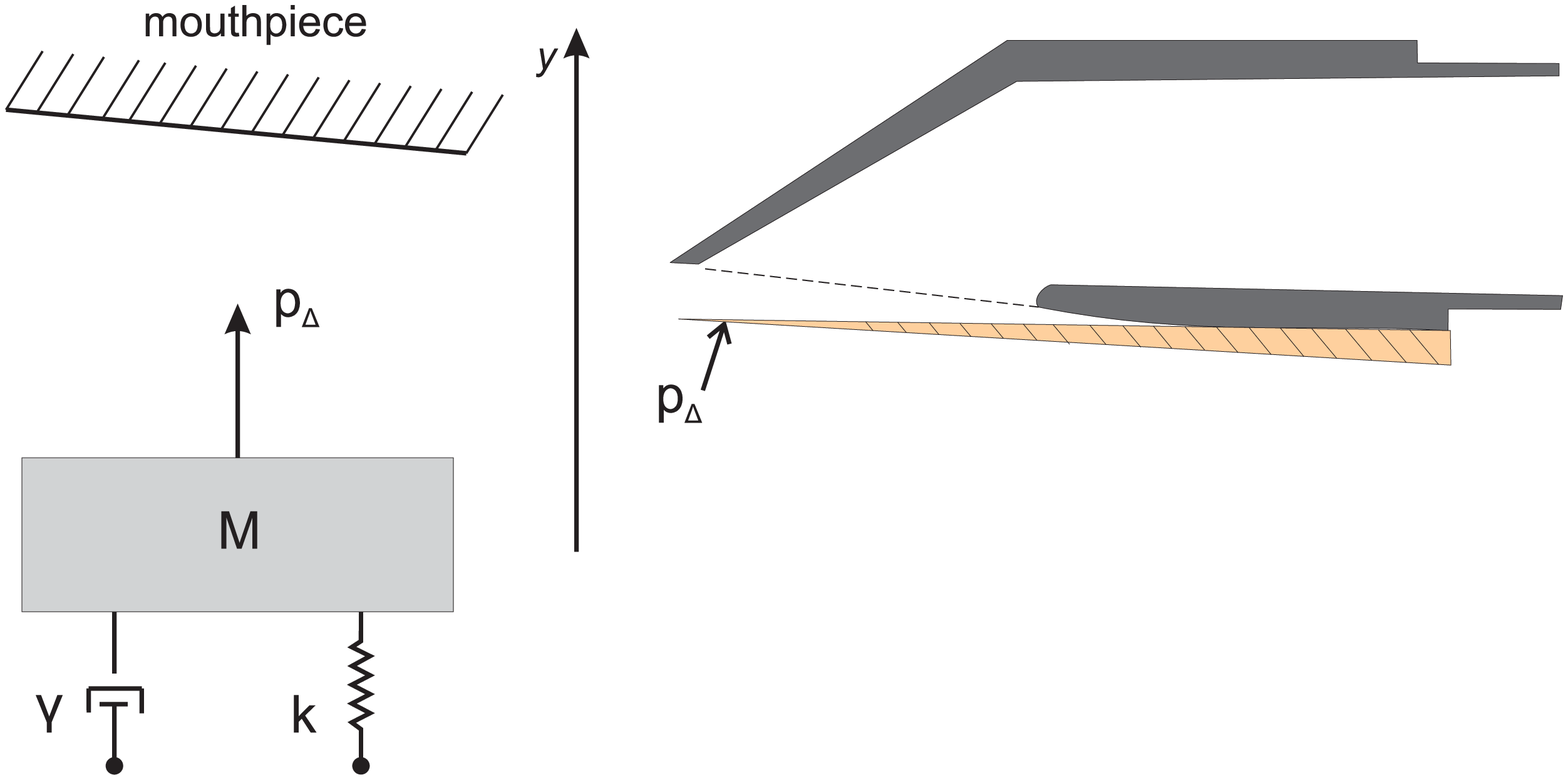}\\
			\vspace{0.1cm}
			(c)\\
			\vspace{-0.1cm}
			\includegraphics[width=1\linewidth]{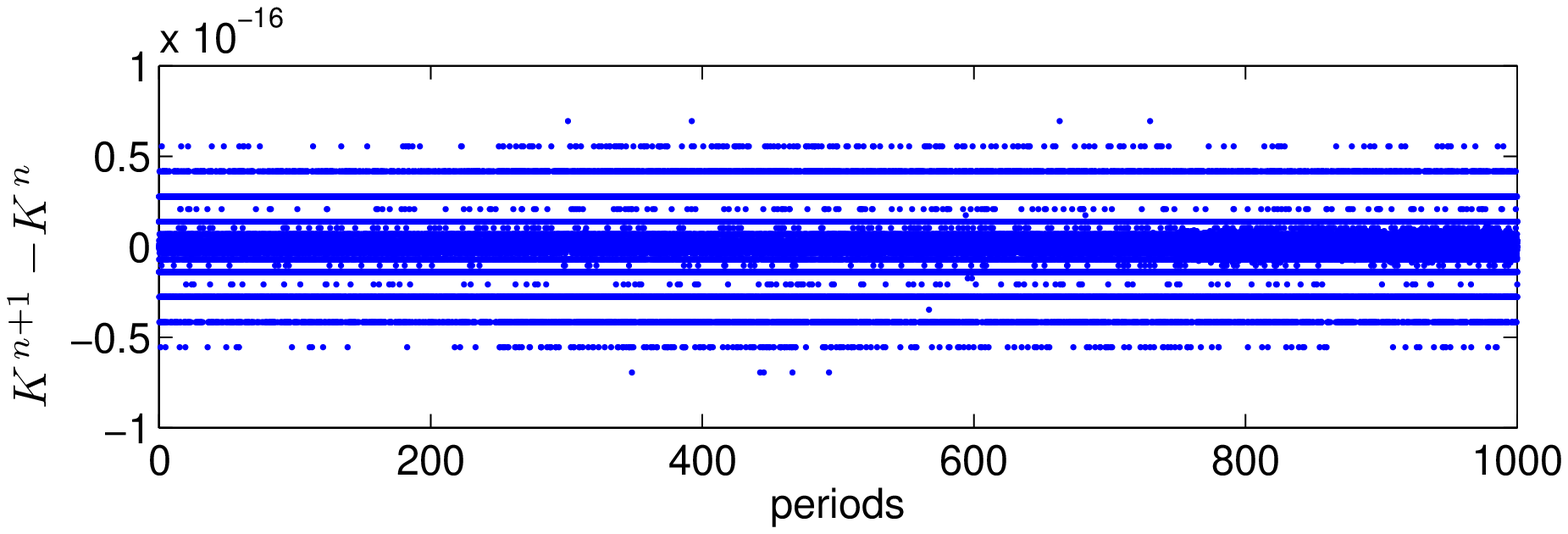}
		\end{center}
	\end{minipage}
	\caption{(a): Displacement signal for a clarinet reed driven by the pressure difference $p_\Delta$ and the respective error in the conservation law (\ref{dHcond}). (b): Lumped model of a clarinet reed with mass per unit area $M$, stiffness per unit area $k = M\omega_0^2$ and damping $\gamma$. (c): Energy error during a longer time interval simulation. The reed parameters are $\omega_0 = 5000\pi$~Hz, $\gamma = 2800$~s$^{-1}$, $y_c = 2.4\times10^{-4}$~m, $M = 0.05$~kg/m$^2$, $k_c = 10^{12}$~Pa/m$^\alpha$ and $\alpha = 1.5$.}
	\label{fig;ddho_NL}
\end{figure}

If this is considered given in the form of a time series $p_\Delta(t)$, it is possible to sample it and calculate the force driving the reed, in order to simulate its oscillations. Defining $M$ as the mass per unit area of the reed, the equation of motion for a lumped reed model becomes
\begin{equation}\label{reedmot}
M\dfrac{\textrm{d}^2y}{\textrm{d}t^2} + M\gamma\dfrac{\textrm{d}y}{\textrm{d}t} + M\omega_0^2 y + k_c\lfloor y-y_c\rfloor^\alpha= p_\Delta
\end{equation}
where $k_c$ is now defined as contact stiffness per unit area and $y_c$ is the point after which the reed-mouthpiece interaction becomes significant \cite{aca12}. The driving force per unit area corresponds to the pressure difference across the reed $p_\Delta$. The results of such a simulation, implemented by solving (\ref{ddhoHeq}) at each time step, using the (EC) method, are shown in Figure \ref{fig;ddho_NL}(a). The pressure difference $p_\Delta$ is synthesised taking into account a typical clarinet spectrum at 146 Hz (Note D3) \cite{fr}, by defining the amplitudes of the first seven harmonics as $\{A_1, A_2,\ldots A_7\} = \{2000, 40, 400, 40, 100, 40, 28\}$ N/m$^2$. The reed parameters are given in the figure caption, the audio sampling rate is used ($f_s = 44.1$ kHz) and rest initial conditions are assumed. Note that, due to zero initial conditions, $K^0 = 0$ and the energy error is defined as $K^{n+1} - K^n$.

In this case, the steady power input from the external force (due to $p_\Delta$) allows the simulation of longer time intervals, without the oscillations dying out, as was the case in the previous section. The energy error for simulating 1000 periods (equivalent to 302085 samples in 6.85 sec for this system, an unusually long duration for musical tones) is shown in Figure \ref{fig;ddho_NL}(c). It can be observed that the conserved quantity $K$ remains constant, within machine precision. This, apart from bounding the solutions of the numerical approximation (and thus ensuring numerical stability) also shows that no `artificial' energy is fed into or lost from the system, besides that due to the external driving and frictional forces.

\section{Extension to distributed systems}\label{sec;dist}

The methodology presented above can be extended to distributed systems by considering infinite dimensional dynamical systems. To this cause the vibration of an ideal string bouncing on a rigid obstacle is considered. The collision force is thus nonlinear, and linear damping is also added to the string model. A similar conservation law, like the one from Section~\ref{sec;lumped} in the lumped case, is obtained for this system. The numerical discretisation is carried out using the (CK) method, which is usually used in lumped models. In a (lossless) Hamiltonian framework the (EC) method has been used to simulate a stiff string \cite{jsv} and also extended to the lossy case \cite{isma14}, whereas discrete gradients \cite{chabassier_cmame} and finite difference methods \cite{stefan_strings} have been used to model nonlinear strings.

\subsection{Lagrangian formulation}\label{sec;lag}
Let a stiff string of length $l$, simply supported at both ends and with given initial displacement $y(x)$ interact with a flat, rigid barrier located below it at height $y_\mathrm{b}$. The Lagrangian density of this system is given by the difference between the kinetic and potential energy density, as
\begin{equation}
\mathcal{L} = \mathcal{T} - \mathcal{V} = \rho A(\partial_t y)^2/2 - \Big(\mathcal{V}_\mathrm{\tau} + \mathcal{V}_\mathrm{s} + \mathcal{V}_\mathrm{b}\Big)
\end{equation}
where $\rho$ is the mass density and $A$ the cross-sectional area of the string.
\begin{equation}
\mathcal{V}_\mathrm{\tau} = \tau(\partial_x y)^2/2,\qquad \mathcal{V}_\mathrm{s} = EI(\partial_{xx} y)^2/2\qquad\textrm{and}\qquad \mathcal{V}_\mathrm{b} = k_\mathrm{b}\lfloor(y_\mathrm{b} - y)^{\alpha+1}\rfloor/(\alpha+1)
\end{equation}
are respectively the potential energies due to the string tension $\tau$, string stiffness $EI$ and the collision potential due to interaction with the barrier. Hence the Lagrangian density is a function of the displacement variable $y(x,t)$ and its space-time derivatives, given by
\begin{equation}\label{Ldensity}
\mathcal{L} = \frac{\rho A}{2}y_t^2 - \frac{\tau}{2}y_x^2 - \frac{EI}{2}y_{xx}^2 - \frac{k_\mathrm{b}\lfloor(y_\mathrm{b} - y)^{\alpha+1}\rfloor}{\alpha+1}
\end{equation}
and the Lagrangian of the system is
\begin{equation}\label{L}
L = \int_0^l\mathcal{L}(y,y_t,y_x,y_{xx};x,t)\,dx
\end{equation}
where the following notation is adopted:
\begin{equation}\label{def_x}
y_t = \partial_t y = \partial y/\partial t,\qquad
y_x = \partial_x y = \partial y/\partial x,\qquad
y_{xx} = \partial_{xx} y = \partial^2 y/\partial x^2.
\end{equation}
The variation of the Lagrangian density subject to a virtual displacement $\delta y$ is
\begin{equation}\label{deltaLdensity}
\delta\mathcal{L} = \delta y \frac{\partial \mathcal{L} }{\partial y} + \frac{\partial\delta y}{\partial t}\frac{\partial \mathcal{L} }{\partial y_t} + \frac{\partial\delta y}{\partial x}\frac{\partial \mathcal{L} }{\partial y_x} +\frac{\partial^2\delta y}{\partial x^2}\frac{\partial \mathcal{L} }{\partial y_{xx}}
\end{equation}
and Hamilton's principle of least action \cite{arnold78} dictates that
\begin{equation}\label{Hprince}
\delta\int L\,dt=0~~~\Rightarrow~~~\delta\iint\mathcal{L}\:dx\,dt = 0.
\end{equation}
Substituting the expression in equation~(\ref{deltaLdensity}) and using integration by parts, along with the fact that $\delta y$ vanishes at the integration boundaries \cite{lanczos}, yields
\begin{equation}\label{intEL}
\iint \delta y\left(\frac{\partial \mathcal{L}}{\partial y} - \frac{\partial}{\partial t}\frac{\partial \mathcal{L}}{\partial y_t} - \frac{\partial}{\partial x}\frac{\partial \mathcal{L}}{\partial y_x} +\frac{\partial^2}{\partial x^2}\frac{\partial \mathcal{L}}{\partial y_{xx}}\right)dx\,dt = 0.
\end{equation} 
The requirement for the integral to be zero for an arbitrary variation $\delta y$ results in the Euler-Lagrange equation for the given dynamical system, which is a partial differential equation of the form
\begin{equation}\label{EL}
\frac{\partial \mathcal{L}}{\partial y} = \frac{\partial}{\partial t}\left(\frac{\partial \mathcal{L}}{\partial y_t}\right) + \frac{\partial}{\partial x}\left(\frac{\partial \mathcal{L}}{\partial y_x}\right) -\frac{\partial^2}{\partial x^2}\left(\frac{\partial \mathcal{L}}{\partial y_{xx}}\right).
\end{equation}
The Hamiltonian density $\mathcal{H}$ can be obtained by defining the conjugate momentum
\begin{equation}\label{p}
p=\frac{\partial \mathcal{L}}{\partial y_t} = \rho Ay_t
\end{equation}
and taking the Legendre transformation of the Lagrangian density
\begin{equation}\label{legendre}
\begin{split}
\mathcal{H} = y_tp-\mathcal{L}(y,y_t,y_x,y_{xx})
&=\frac{1}{2}\frac{p^2}{\rho A} + \frac{1}{2}\tau y_x^2 + \frac{1}{2}EIy_{xx}^2 +\frac{k_\mathrm{b}}{\alpha+1}\lfloor(y_\mathrm{b}-y)^{\alpha+1}\rfloor\\[0.2cm]
&= \mathcal{T}(p) + \mathcal{V}_\mathrm{\tau}(y_x) + \mathcal{V}_\mathrm{s}(y_{xx}) + \mathcal{V}_\mathrm{b}(y).
\end{split}
\end{equation}
Equations (\ref{EL}), (\ref{p}) and (\ref{legendre}) can be combined to formulate Hamilton's equations of motion, which with the inclusion of a linear damping term are:
\begin{subequations}\label{Heq}
	\begin{align}
		\frac{\partial p}{\partial t} &= \frac{\partial}{\partial x}\left(\frac{\partial\mathcal{H}}{\partial y_x}\right) - \frac{\partial^2}{\partial x^2}\left(\frac{\partial\mathcal{H}}{\partial y_{xx}}\right) - \frac{\partial\mathcal{H}}{\partial y} -\gamma p\label{Heq1}\\
		\frac{\partial y}{\partial t} &= \frac{\partial\mathcal{H}}{\partial p}.\label{Heq2}
	\end{align}
\end{subequations}
The Hamiltonian (total energy) of the system is computed as
\begin{equation}\label{Hint}
H = \int_0^l\mathcal{H}(y,p,y_x,y_{xx})\;dx
\end{equation}
with the conserved quantity being equal to
\begin{equation}\label{Hdistcon}
\int_0^l\left(\mathcal{H}+\int\dfrac{\gamma p^2}{m}\;dt\right)\;dx = \mbox{const}.
\end{equation}
In what follows, an ideal string will be considered (by setting $EI = 0$, which results in $\mathcal{V}_\mathrm{s} = 0$) so that a comparison with an analytic result will be possible.

\subsection{Caldirola-Kanai formalism and discretisation}\label{sec;CKdisc}

Transforming to the Caldirola-Kanai formalism the following Hamiltonian density is defined
\begin{equation}\label{HCKdist}
\mathcal{H}_\mathrm{CK} = e^{-\gamma t}\frac{\varpi^2}{2m} + e^{\gamma t}\left(\frac{1}{2}\tau y_x^2+\frac{k_\mathrm{b}}{\alpha+1}\lfloor(y_\mathrm{b}-y)^{\alpha+1}\rfloor\right)
\end{equation}
where $\varpi = e^{\gamma t}m\dot y$, with
\begin{equation}\label{HeqCKd}
\dfrac{\textrm{d}y}{\textrm{d}t}=\dfrac{\partial \mathcal{H}_{CK}}{\partial \varpi},\qquad \dfrac{\textrm{d}\varpi}{\textrm{d}t}=-\dfrac{\partial \mathcal{H}_{CK}}{\partial y}.
\end{equation}
Space and time discretisation is carried out by denoting the value of variable $y$ at position $x = m\Delta x$ and time $t = n\Delta t$ by $y_m^n$, $\Delta x$ being the spatial sampling interval (see Figure \ref{fig;grid}):
\begin{subequations}\label{stringHdisc}
	\begin{align}
		\frac{y_m^{n+1}-y_m^n}{\Delta t}= e^{-\gamma (n+\frac{1}{2})\Delta t}&\dfrac{\varpi_m^{n+1} - \varpi_m^n}{2\rho A}\label{stringHdisc1}\\[0.2cm]
		\begin{split}
			\frac{\varpi_m^{n+1}-\varpi_m^n}{\Delta t} =e^{\gamma (n+\frac{1}{2})\Delta t}&\Bigg(\dfrac{\tau}{2}\delta_\Delta(y_m^{n+1}+y_m^n)\\
			&- \frac{k_b}{\alpha+1}\dfrac{\lfloor(y_{b}-y_m^{n+1})^{\alpha+1}\rfloor - \lfloor(y_{b}-y_m^n)^{\alpha+1}\rfloor}{y_m^{n+1}-y_m^n}\Bigg)
		\end{split}
	\end{align}
\end{subequations}
where $\delta_\Delta(y_m^n) = \delta_{x+}\delta_{x-} y_m^n$, with
\begin{equation}\label{defdelta}
\delta_{x+}y_m^n = \frac{y_{m+1}^n-y_m^n}{\Delta x},\qquad\delta_{x-}y_m^n = \frac{y_m^n-y_{m-1}^n}{\Delta x}.
\end{equation}
\begin{figure}[!t]
	\centering
	\includegraphics[width=0.5\columnwidth]{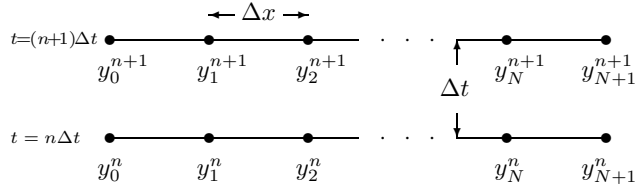}
	\caption{The discretised string at times $n\Delta t$ and $(n+1)\Delta t$.\label{fig;grid}}
\end{figure}
Scheme (\ref{stringHdisc}) is centred at time $t=(n+1/2)\Delta t$ and position $x=m\Delta x$. Defining a `normalised momentum' variable $q_m^n = \varpi_m^ne^{-\gamma n\Delta t}\Delta t/(2\rho A)$ it can be written in matrix form as
\begin{subequations}\label{stringM}
	\begin{align}
		&\mathbf{y}^{n+1}- \mathbf{y}^n =  r\mathbf{q}^{n+1} +  r^{-1}\mathbf{q}^n\\[0.2cm]
		& r\mathbf{q}^{n+1} -  r^{-1}\mathbf{q}^n = \beta_2 \mathbf{D}_2\left( \mathbf{y}^{n+1} +  \mathbf{y}^n\right) -\zeta \mathbf{S}^{-1}\left(\lfloor( \mathbf{y_b}- \mathbf{y}^{n+1})^{\alpha+1}\rfloor - \lfloor( \mathbf{y_b}- \mathbf{y}^n)^{\alpha+1}\rfloor\right)
	\end{align}
\end{subequations}
where $\mathbf{S}=\diag( \mathbf{y}^{n+1} -  \mathbf{y}^n)$ is a diagonal matrix,
\begin{equation}\label{phidef}
r = e^{\gamma\Delta t/2},\qquad\beta_2 = \frac{\tau\Delta t^2}{4\rho A\Delta x^2},\qquad\zeta=\frac{k_b\Delta t^2}{2\rho A(\alpha+1)}
\end{equation}
and $ \mathbf{y}^n$, $ \mathbf{y_b}^n$ and $ \mathbf{q}^n$ are column vectors holding displacement, barrier profile and normalised momentum values. Under the assumption of simply supported boundary conditions on both ends of the system, these vectors hold the values of $N$ interior nodes on the string (i.e.\, from $y_1$ to $y_{N}$), and $ \mathbf{D}_2$ is then an $N\times N$ tridiagonal matrix:
\begin{equation}\label{D2}
\mathbf{D}_2 = \left[\begin{array}{cccc}
-2 & 1 & & 0\\
1 & \ddots & \ddots & \\
& \ddots & \ddots & 1 \\
0 & & 1 & -2 \end{array} \right]
\end{equation}
which implements the second spatial derivative of the string state. For alternative types of boundary conditions see \cite{jsv}, where a similar discretisation approach is carried out on a conservative system using the (EC) method. Singularities in the diagonal matrix $ \mathbf{S}$ can be avoided by considering only the vibrating portion of the string (with $y_m^{n+1} \neq y_m^n$; otherwise $y_m^{n+1}$ is known and $\varpi_m^{n+1}$ can be obtained from (\ref{stringHdisc1})). Now setting
\begin{equation}\label{stringx}
\mathbf{s}= \mathbf{y}^{n+1}- \mathbf{y}^n= r\mathbf{q}^{n+1}+ r^{-1}\mathbf{q}^n
\end{equation}
yields the nonlinear system of equations
\begin{equation}\label{Fstring}
\begin{split}
\mathbf{F} = &\left( \mathbf{I} -  \beta_2\mathbf{D}_2\right) \mathbf{s} - 2\left( \beta_2\mathbf{D}_2 \mathbf{y}^n +  r^{-1}\mathbf{q}^n\right)\\[0.2cm]
{} &+\zeta \mathbf{S}^{-1} \left(\lfloor( \mathbf{y_b}- \mathbf{y}^n- \mathbf{s})^{\alpha+1}\rfloor - \lfloor( \mathbf{y_b}- \mathbf{y}^n)^{\alpha+1}\rfloor\right) =  \mathbf{0}.
\end{split}
\end{equation}
This can be solved for $\mathbf{s}$ using the multidimensional Newton method. The Jacobian of $ \mathbf{F}$ is
\begin{equation}\label{jacobian}
\mathbf{J} =  \mathbf{I} - \beta_2\mathbf{D}_2 + \mathbf{C}
\end{equation}
where $ \mathbf{C}$ is a diagonal matrix with elements
\begin{equation}\label{cij}
\{c_{i,i}\} = \frac{\Delta t^2}{2\rho A}\frac{s_i\,\mathcal{V}_b'(y_i^n+s_i) - \mathcal{V}_b(y_i^n+s_i) + \mathcal{V}_b(y_i^n)}{s_i^2}
\end{equation}
where $\mathcal{V}_b'$ signifies taking the derivative of $\mathcal{V}_b$ with respect to displacement. It can be shown that $\mathbf{J}$ is positive definite \cite{jsv}, which ensures the uniqueness of a root of equation~(\ref{Fstring}) \cite{deuflhard}. Furthermore $\mathbf{J}$ is also an $M$-matrix, which guarantees global convergence of the Newton method for finding the roots of the componentwise convex function $ \mathbf{F}$ \cite{ortega}.

Solving (\ref{Fstring}) $\mathbf{y}^{n+1}$ and $\mathbf{q}^{n+1}$ can be updated as
\begin{equation}\label{updates}
\mathbf{y}^{n+1} = \mathbf{y}^n + \mathbf{s},\qquad\mathbf{q}^{n+1} = r^{-1}\mathbf{s} - r^{-2}\mathbf{q}^n.
\end{equation}

In accordance to the lumped case, the energy density $\mathcal{H}$ can be calculated from the Hamiltonian density $\mathcal{H}_\mathrm{CK}$ as $\mathcal{H} = e^{-\gamma t}\mathcal{H}_\mathrm{CK}$ and the total energy of the system is given by integration along the length of the string \cite{jsv}:
\begin{equation}\label{numenergy}
H^n = b\Big[ (\mathbf{q}^n)^t \mathbf{q}^n - (\mathbf{y}^n)^t \beta_2\mathbf{D}_2 \mathbf{y}^n + \zeta \mathbf{1}^t\lfloor( \mathbf{y_b}- \mathbf{y}^n)^{\alpha+1}\rfloor\Big]
\end{equation}
with $ \mathbf{1}=(1,\ldots,1)^t$ and $b=2\rho A\Delta x/\Delta t^2$.
Furthermore, from the definition of the Hamiltonian system (\ref{HeqCKd}), it follows that
\begin{equation}\label{ECK}
\frac{\mathrm{d}\mathcal{H}}{\mathrm{d} t} = \frac{\mathrm{d}}{\mathrm{d} t}\left(e^{-\gamma t}\mathcal{H_\mathrm{CK}}\right) = -\frac{4\gamma \rho A}{\Delta t^2}q^2,\qquad \mathrm{with}\quad q = \varpi e^{-\gamma t}\Delta t/(2\rho A) 
\end{equation}
and the discrete conserved quantity becomes
\begin{equation}\label{Kstring}
K^{n} = H^{n+1} + \sum_{\kappa=0}^n 2b\gamma\Delta t(\mathbf{q}^n)^t \mathbf{q}^n = \mbox{const}.
\end{equation}

In the absence of losses (for $\gamma = 0$) an analytical result \cite{cabannes} states that if a straight obstacle is placed at half the amplitude of the string vibration, then the period of the oscillation will become 1.5 times larger \cite{han_cable}. This is verified in Figure \ref{fig;cab}(a), where for increasing sampling rates this behaviour is reproduced. A 0.7 meter long string is simulated, with tension $\tau = 100$ N, and linear mass density $\rho A = 0.001$ kg$\,$m$^{-1}$; $k_b = 10^7$ and $\alpha = 1$ are used to model a rigid obstacle and $\Delta x = 0.007$ m. Figure~\ref{fig;cab}(c) shows the error in the conservation of energy (which is constant in this case) defined as $e^n = (H^n - H^0)/H^0$. On the right hand side of the same figure it is shown how including losses affects the system (with $\gamma = 200$ s$^{-1}$). Figure~\ref{fig;cab}(b) depicts the mid-point displacement of the string (calculated using a sampling rate of 44.1 kHz) and Figure \ref{fig;cab}(d) plots the energy components and the conserved quantity $K$.

\begin{figure}[!t]
	\centering
	\includegraphics[width=0.99\columnwidth]{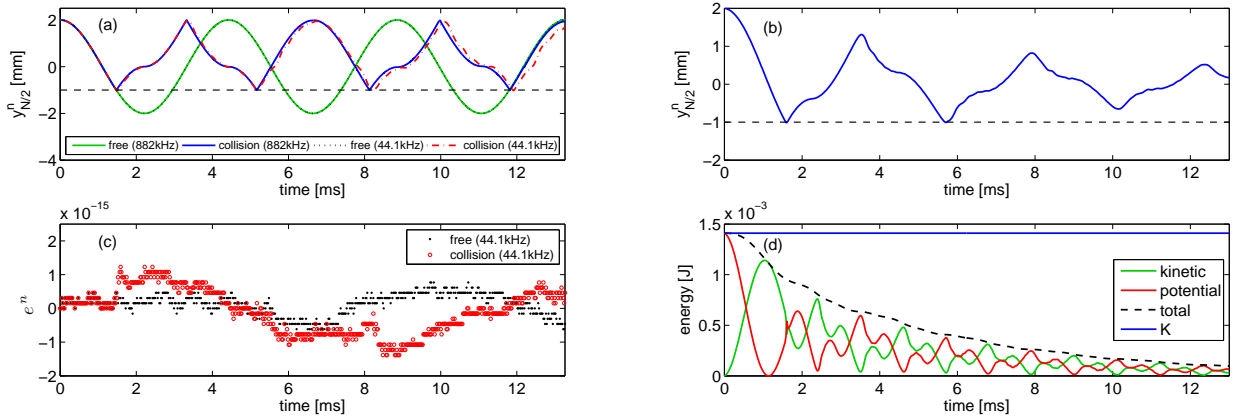}
	\caption{Simulation of an ideal string bouncing on a flat, rigid obstacle (located at $y_b = -0.0001$ m), with initial conditions $y(x,0) = 0.0002\sin(\pi x/l)$. (a): Mid-point displacement of a lossless impeded string, compared to a free vibrating string. (b): Mid-point displacement of a damped string. (c): Error in the conservation of energy for the lossless string. (d): Energy components and the conserved quantity $K$ for the lossy string.\label{fig;cab}}
\end{figure}

Finally, for the lossless case it is possible to run the simulation for longer time intervals, without the oscillations decaying to zero. Figure \ref{fig;long10_per} depicts the energy error in terms of both the deviation from the initial energy as well as the error per time-step for a total of 100 seconds (14500 periods of oscillation) with $f_s=44.1$ kHz. A small energy drift is observed in the top plot, as discussed in Section~\ref{sec;nl}. However, in the presence of losses, such a drift (whether adding to or removing energy from the system) is negligible compared to the loss mechanism caused by damping effects.

\begin{figure}[h]
	\centering
	\includegraphics[width=0.6\columnwidth]{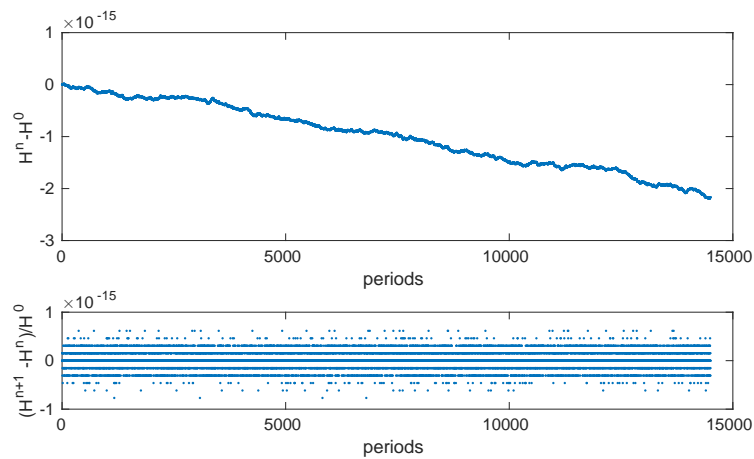}
	\caption{Simulation of a lossless vibrating string. Top: Energy error. Bottom: Normalised error per time-step.\label{fig;long10_per}}
\end{figure}

\section{Conclusion}\label{sec;conc}

This paper has provided an analysis of discretisation algorithms for non-conservative acoustic systems from the perspective of energy-conserving and symplectic schemes. Existing schemes have been analysed and novel ones have been formulated in an attempt to respect the contraction of the symplectic structure. It has been shown that existing methods, such as the impulse invariance method (used in digital signal processing) and the Caldirola-Kanai approach (used in quantum mechanics) offer an exact discretisation of the symplectic form of the system. Mechanical integrators sharing the same property have been also generated from other Hamiltonian integrators using splitting methods. The accuracy of a series of algorithms has been assessed in comparison with an analytic solution for a linear oscillator.

When nonlinear forces are acting on the system, numerical measures may be used to evaluate the numerical schemes. This has been exemplified for a lumped collision model. Similar to energy loss due to damping, it has been shown how power input due to external forces may be incorporated to the presented analysis. As a case study the motion of a clarinet reed has been simulated, with the output respecting the underlying discrete conservation law. Finally it has been shown how to extend the proposed methodology to distributed systems by presenting an application of the Caldirola-Kanai method to the simulation of a vibrating string involving a nonlinear interaction. The numerical simulations were able to reproduce analytical results available for a particular setting, when a rigid barrier is placed halfway across the amplitude of a lossless vibrating string.

The presented analysis naturally extends to other nonlinear interactions that take place in musical instruments, such as the piano hammer-string interaction and mallet impacts on a membrane. Furthermore, in the area of speech synthesis, modelling the collision of the vocal folds, where impact damping is expected to be rather high, could also be treated using the proposed approach.



\bibliographystyle{plain}
\bibliography{chatzi20}

\end{document}